\documentclass[letterpaper, 10 pt, conference]{ieeeconf}  

\IEEEoverridecommandlockouts                              

\usepackage{amsmath,amssymb,amsfonts}

\usepackage{amsthm}
\usepackage{algorithmic}
\usepackage{graphicx}
\usepackage{textcomp}
\usepackage{xcolor}
\usepackage{bm}
\usepackage{mathtools}
\usepackage{booktabs}

\newtheorem{theorem}{Theorem}
\newtheorem{proposition}{Proposition}
\newtheorem{lemma}{Lemma}
\newtheorem{assumption}{Assumption}
\newtheorem{definition}{Definition}
\newtheorem{remark}{Remark}
\newtheorem{corollary}{Corollary}

\title{\LARGE \bf
Data-Attributed Adaptive Control Barrier Functions:\\
Safety-Certified Training Data Curation via Influence Analysis
}

\author{
Jiachen Li$^{1}$, Shihao Li$^{1}$, and Dongmei Chen$^{1}$%
\thanks{$^{1}$All authors are with the Department of Mechanical Engineering, The University of Texas at Austin, Austin, TX 78712, USA.
        {\tt\small \{jiachenli, shihaoli01301, dmchen\}@utexas.edu}}%
}

\begin{document}

\maketitle
\thispagestyle{empty}
\pagestyle{empty}

\begin{abstract}
Learning-based adaptation of Control Barrier Function (CBF) parameters offers a promising path toward safe autonomous navigation that balances conservatism with performance. Yet the accuracy of the underlying safety predictor is ultimately constrained by training data quality, and no prior work has formally characterized how prediction errors propagate through the adaptive pipeline to degrade closed-loop safety guarantees. We introduce Data-Attributed Adaptive CBF (DA-CBF), a framework that integrates TracIn-based data attribution into adaptive CBF learning. Our theoretical contributions are fourfold: (i) corrected two-sided bounds relating the safety-loss surrogate to the CBF constraint margin; (ii) a safety margin preservation theorem showing that prediction error induces quantifiable margin degradation and, via a smooth parameter selector, yields a genuine closed-loop forward invariance guarantee not conditioned on a fixed trajectory; (iii) a CBF-QP constraint perturbation bound that links prediction accuracy directly to recursive feasibility; and (iv) a principled leave-one-out justification for influence-based data curation under explicit smoothness assumptions. On a DynamicUnicycle2D benchmark, DA-CBF reduces prediction RMSE by 35.6\%, expands the certified safe operating set by 39\%, and achieves collision-free navigation in a 16-obstacle environment where the uncurated baseline incurs 3 collisions.
\end{abstract}

\section{INTRODUCTION}

Safety remains a defining challenge for autonomous systems operating in complex, unstructured environments \cite{amodei2016concrete}. As mobile robots are increasingly deployed in settings ranging from cluttered warehouses to urban streets, the need for rigorous, real-time safety assurances has grown correspondingly urgent. Control Barrier Functions (CBFs) \cite{ames2017cbf} offer an appealing framework for meeting this need: they enforce forward invariance of a designated safe set by imposing a pointwise constraint on the control input, one that can be embedded in a Quadratic Program (QP) and solved in real time.

\begin{figure*}[t]
    \centering
    \includegraphics[width=\textwidth]{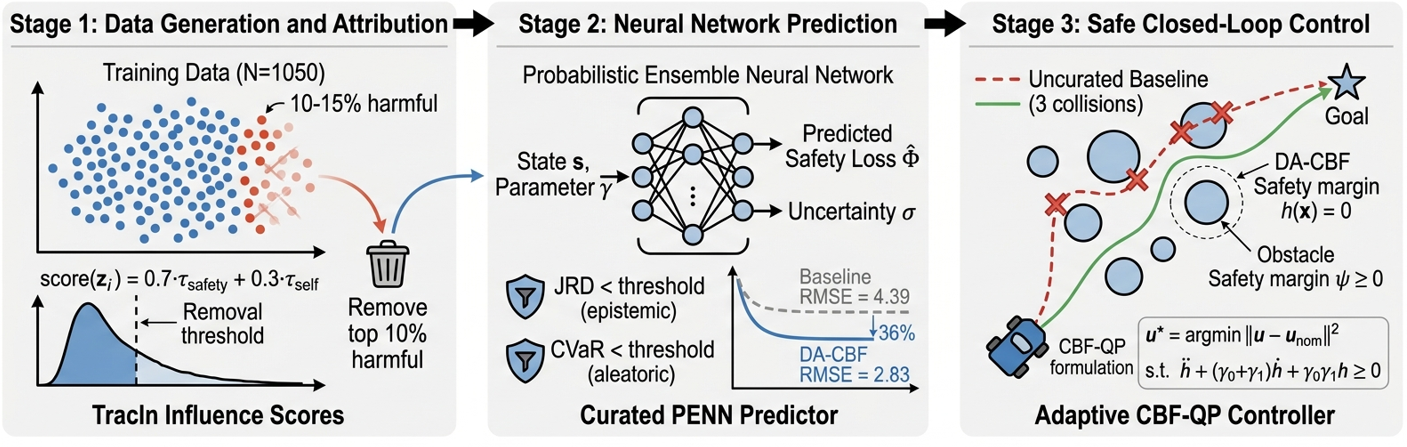}
    \caption{Overview of the DA-CBF pipeline. \textbf{Stage~1:} TracIn-based attribution identifies and removes the top 10\% harmful training samples. \textbf{Stage~2:} Curated PENN predicts safety loss with 36\% lower RMSE, filtered by JRD and CVaR thresholds. \textbf{Stage~3:} Adaptive CBF-QP achieves collision-free navigation where the baseline incurs 3 collisions.}
    \label{fig:overview}
\end{figure*}

In practice, however, the effectiveness of CBF controllers depends heavily on the class-$\mathcal{K}$ function parameters that mediate the trade-off between conservatism and performance. Overly conservative parameters keep the system safe but produce sluggish, inefficient behavior; overly aggressive ones improve task performance at the cost of potential constraint violations. Extensions to higher-order CBFs \cite{xiao2021hocbf} and tunable input-to-state-safe formulations \cite{alan2022tunable} enrich the design space but do not resolve the underlying parameter selection problem, leaving open the question of how to choose parameters adaptively in response to the current operating context.

Kim et al.~\cite{kim2025adaptive} address this gap with an online adaptation scheme built on a Probabilistic Ensemble Neural Network (PENN) \cite{lakshminarayanan2017simple}. Their pipeline predicts a safety-loss surrogate for each candidate parameter setting, then applies a two-stage uncertainty filter: epistemic uncertainty is assessed via Jensen--R\'{e}nyi Divergence \cite{renyi1961measures} to flag out-of-distribution queries, while aleatoric risk is bounded through Conditional Value-at-Risk (CVaR). Only candidates that survive both filters are considered, and the least conservative among them is selected for control.

The resulting architecture is carefully designed, yet its end-to-end safety guarantee can be no stronger than the PENN's predictive accuracy, and that accuracy is, in turn, shaped by the training data. Data for adaptive CBF learning is typically generated through large-scale simulation sweeps over initial conditions, obstacle layouts, and candidate parameters. Such sweeps inevitably introduce noisy labels from near-boundary dynamics, distributional imbalances that over-represent benign configurations, and outliers arising from numerical edge cases. The effect of these corrupted samples is not uniform: they bias predictions most severely in precisely the high-risk regions of state space where accurate safety-loss estimates matter most.

A growing body of work in data attribution, spanning influence functions \cite{koh2017influence}, TracIn \cite{pruthi2020tracin}, TRAK \cite{park2023trak}, and memorization-based diagnostics \cite{feldman2020neural}, provides tools for identifying which training examples drive specific prediction errors. These methods trace the effect of individual samples along the optimization trajectory, yielding principled per-sample scores. What remains absent, however, is \emph{a formal account of how data quality connects to closed-loop safety in adaptive CBF pipelines}. Without such an account, one cannot determine how much prediction improvement is needed to recover a target safety margin, and data curation reduces to heuristic trial and error rather than a certifiable intervention.

The broader literature on learning-based safety-critical control has not filled this gap. Taylor et al.~\cite{taylor2020learning} learn CBF parameters from expert demonstrations but assume clean training data. Castaneda et al.~\cite{castaneda2021gp} use Gaussian processes to handle model uncertainty in CBF feasibility, yet do not examine how training data quality affects the GP posterior. Choi et al.~\cite{choi2020reinforcement} combine reinforcement learning with control Lyapunov and barrier functions under model uncertainty, relying on representative state coverage without formal data quality analysis. Dawson et al.~\cite{dawson2023safe} survey neural Lyapunov, barrier, and contraction methods, noting that data-driven certificate learning generally assumes sufficient training coverage. In each case, prediction error is treated as either negligible or absorbed by conservative margins, without analyzing how data deficiencies degrade the resulting safety certificate. Safe MPC approaches \cite{zeng2021mpc} and onboard CBF implementations for high-speed platforms \cite{singletary2022onboard} make analogous modeling assumptions. The upshot is that practitioners currently lack a principled basis for deciding whether their training data is adequate for a given safety target, or for determining how best to improve it.

We propose Data-Attributed Adaptive CBF (DA-CBF), whose three-stage pipeline is illustrated in Fig.~\ref{fig:overview}. Our work makes four contributions. First, we prove a closed-loop safety margin preservation theorem: Theorem~\ref{thm:safety_degradation} shows that a learned controller inherits the safety of an oracle adaptive controller with certified margin $\delta_{\min}$ whenever the prediction error satisfies $\epsilon \leq \delta_{\min}/(L_\psi L_\mathcal{M})$, with a sampled-data extension and a probabilistic guarantee via a covering argument in Corollary~\ref{cor:prob_safety}. Second, we derive a CBF-QP constraint perturbation bound in Proposition~\ref{prop:qp_feasibility} with analytic Lipschitz constants, tying prediction accuracy directly to recursive feasibility. Third, we establish corrected two-sided $\Phi \leftrightarrow \psi$ bounds in Proposition~\ref{prop:phi_bound} and design a smooth parameter selector with a proven Lipschitz constant in Lemma~\ref{lem:lipschitz_selection}, replacing the switching discontinuity of the original hard-threshold selector. Fourth, we supply a principled leave-one-out motivation for TracIn-based data curation in Proposition~\ref{thm:curation} and show experimentally that DA-CBF reduces RMSE by 35.6\% and achieves collision-free navigation in a challenging multi-obstacle environment.

\section{PRELIMINARIES AND PROBLEM FORMULATION}

\subsection{System Dynamics and CBF}

Consider an affine control system:
\begin{equation}\label{eq:dynamics}
\dot{x} = f(x) + g(x)u, \quad x \in \mathcal{X} \subset \mathbb{R}^n,\; u \in \mathcal{U} \subset \mathbb{R}^m
\end{equation}
A set $\mathcal{S} = \{x \in \mathbb{R}^n : h(x) \geq 0\}$ is forward invariant under CBF control \cite{ames2017cbf,nagumo1942lage} if there exists an extended class-$\mathcal{K}_\infty$ function $\alpha$ such that:
\begin{equation}\label{eq:cbf}
\sup_{u \in \mathcal{U}} \left[ L_f h(x) + L_g h(x) u + \alpha(h(x)) \right] \geq 0, \quad \forall x \in \mathcal{S}
\end{equation}
For relative-degree-2 systems, the CBF condition takes the form:
\begin{equation}\label{eq:rd2}
\ddot{h} + (\gamma_0 + \gamma_1)\dot{h} + \gamma_0 \gamma_1 h \geq 0
\end{equation}
where $\gamma_0, \gamma_1 > 0$ parameterize the class-$\mathcal{K}$ function. The corresponding CBF-QP selects:
\begin{equation}\label{eq:cbfqp}
u^*(x, \gamma) = \arg\min_{u \in \mathcal{U}} \|u - u_{\text{nom}}\|^2 \quad \text{s.t. } \eqref{eq:rd2}
\end{equation}

Throughout, we evaluate on the DynamicUnicycle2D model with state $(x, y, \theta, v)$ and control inputs consisting of acceleration and angular velocity.

\subsection{Adaptive CBF Pipeline}

The adaptive pipeline of \cite{kim2025adaptive} trains a PENN $\hat{\Phi}_\theta(s, \gamma)$ to predict the safety loss $\Phi$ for state-parameter pairs $(s, \gamma)$, where $s$ encodes relative state (distance, velocity, approach angle). At each control step, the pipeline samples candidate $\gamma$ values, predicts $(\hat{\Phi}, \hat{T}_d)$ with uncertainty estimates, filters by JRD (epistemic) and CVaR (aleatoric) thresholds, and selects the least conservative valid parameters.

\subsection{Safety-Loss Surrogate}

The safety loss function introduced in \cite{kim2025adaptive} is:
\begin{equation}\label{eq:phi}
\Phi(x, x_{\text{obs}}, \gamma) = \frac{\lambda_1 e^{-\lambda_2 \psi}}{\beta_1 e^{-\beta_2(\cos \Delta\theta + 1)} d^2 + 1}
\end{equation}
where $d$ denotes obstacle distance, $\Delta\theta$ is the relative heading angle, $\psi$ is the CBF constraint value from \eqref{eq:rd2}, and $(\lambda_1, \lambda_2, \beta_1, \beta_2)$ are design parameters.

To ground our subsequent analysis, we first establish a formal connection between $\Phi$ and the barrier certificate.

\begin{assumption}[Bounded Operating Domain]\label{asm:domain}
The system operates in a compact domain $\mathcal{X}_{\text{op}}$ with obstacle distance $d \in [d_{\min}, d_{\max}]$ where $d_{\min} > 0$, relative heading $\Delta\theta \in [0, \pi]$, and CBF parameters $\gamma \in \Gamma := [\gamma_{\min}, \gamma_{\max}]^2$ with $\gamma_{\min} > 0$.
\end{assumption}

\begin{assumption}\label{asm:lipschitz}
The CBF constraint margin $\psi(\gamma) := \ddot{h} + (\gamma_0 + \gamma_1)\dot{h} + \gamma_0\gamma_1 h$ is Lipschitz continuous in $\gamma$ with constant $L_\psi > 0$ over $\Gamma$, i.e., $|\psi(\gamma) - \psi(\gamma')| \leq L_\psi \|\gamma - \gamma'\|$ for all $\gamma, \gamma' \in \Gamma$.
\end{assumption}

\begin{remark}
Assumption~\ref{asm:lipschitz} is mild for systems with bounded $\dot{h}$ and $h$ on $\mathcal{X}_{\text{op}}$. Since $\psi$ is polynomial in $\gamma$ with coefficients depending on $\dot{h}$ and $h$, Lipschitz continuity on compact $\Gamma$ follows with $L_\psi \leq |\dot{h}| + |\dot{h}| + 2\gamma_{\max}|h| + |h|$.
\end{remark}

\begin{proposition}[Two-Sided Bounds on Safety-Loss Surrogate]\label{prop:phi_bound}
Under Assumption~\ref{asm:domain}, define the denominator bounds:
\begin{align}
D_{\min} &:= \beta_1 e^{-2\beta_2} d_{\min}^2 + 1 \label{eq:dmin}\\
D_{\max} &:= \beta_1 d_{\max}^2 + 1 \label{eq:dmax}
\end{align}
Then for all $(x, x_{\text{obs}}, \gamma) \in \mathcal{X}_{\text{op}} \times \Gamma$, the safety loss $\Phi$ and CBF margin $\psi$ satisfy:
\begin{equation}\label{eq:two_sided}
\frac{\lambda_1}{D_{\max}} e^{-\lambda_2 \psi} \leq \Phi \leq \frac{\lambda_1}{D_{\min}} e^{-\lambda_2 \psi}
\end{equation}
Inverting \eqref{eq:two_sided} yields a usable error bound depending only on the predicted $\hat{\Phi}$ and domain constants: if $|\hat{\Phi} - \Phi| \leq \epsilon_\Phi$, then
\begin{equation}\label{eq:psi_error}
|\hat{\psi} - \psi| \leq \frac{1}{\lambda_2} \ln\!\left(1 + \frac{D_{\max}}{\lambda_1 e^{-\lambda_2 \hat{\psi}_{\max}}} \epsilon_\Phi \right)
\end{equation}
where $\hat{\psi}_{\max}$ is an upper bound on $\psi$ computable from $\hat{\Phi}$ via the left inequality in \eqref{eq:two_sided}.
\end{proposition}

\begin{proof}
From \eqref{eq:phi}, $\Phi = \lambda_1 e^{-\lambda_2 \psi} / D$ where $D := \beta_1 e^{-\beta_2(\cos\Delta\theta + 1)}d^2 + 1$. On $\mathcal{X}_{\text{op}}$, $e^{-\beta_2(\cos\Delta\theta+1)} \in [e^{-2\beta_2}, 1]$ and $d \in [d_{\min}, d_{\max}]$, so $D_{\min} \leq D \leq D_{\max}$: the minimum occurs at $\cos\Delta\theta = 1$ (head-on, $e^{-2\beta_2}$ smallest) and $d = d_{\min}$, the maximum at $\cos\Delta\theta = -1$ ($e^0 = 1$) and $d = d_{\max}$. Since $\Phi = \lambda_1 e^{-\lambda_2\psi}/D$, larger $D$ decreases $\Phi$ and smaller $D$ increases $\Phi$, giving \eqref{eq:two_sided}.

For \eqref{eq:psi_error}: from $\psi = -\frac{1}{\lambda_2}\ln(\Phi D/\lambda_1)$ and $\hat{\psi} = -\frac{1}{\lambda_2}\ln(\hat{\Phi} D/\lambda_1)$, we get $|\hat{\psi} - \psi| = \frac{1}{\lambda_2}|\ln(\hat{\Phi}/\Phi)| = \frac{1}{\lambda_2}|\ln(1 + (\hat{\Phi}-\Phi)/\Phi)| \leq \frac{1}{\lambda_2}\ln(1 + \epsilon_\Phi/\Phi)$. Using $\Phi \geq (\lambda_1/D_{\max})e^{-\lambda_2\psi}$ and bounding $\psi \leq \hat{\psi}_{\max}$ from the prediction gives \eqref{eq:psi_error}.
\end{proof}

\subsection{Formal Problem Statement}

\begin{definition}[Safety-Weighted Risk]\label{def:risk}
Given a data distribution $\mathcal{P}$ over state-parameter-loss triples $(s, \gamma, \Phi)$, define the safety-weighted generalization risk of predictor $\hat{\Phi}_\theta$:
\begin{equation}\label{eq:risk}
R_{\text{safe}}(\theta) = \mathbb{E}_{(s,\gamma,\Phi) \sim \mathcal{P}} \left[ w(\Phi) \cdot \ell(\hat{\Phi}_\theta(s,\gamma), \Phi) \right]
\end{equation}
where $w(\Phi) = \mathbf{1}[\Phi > \Phi_{\text{thr}}]$ is a safety weight emphasizing high-risk regions, and $\ell$ is the prediction loss.
\end{definition}

\textbf{Problem.} Given a training dataset $\mathcal{D} = \{(s_i, \gamma_i, \Phi_i)\}_{i=1}^N$, identify a subset $\mathcal{D}_{\text{harmful}} \subset \mathcal{D}$ such that retraining on $\mathcal{D} \setminus \mathcal{D}_{\text{harmful}}$ minimizes $R_{\text{safe}}(\theta)$, thereby improving safety-critical prediction accuracy and, through Proposition~\ref{prop:phi_bound}, tightening the bound on CBF constraint margin estimation error.

\section{SAFETY-CERTIFIED DATA CURATION}

\subsection{From Prediction Error to Safety Degradation}

We turn to the central theoretical question: how do prediction errors in $\hat{\Phi}$ translate into degradation of the CBF safety margin?

\begin{assumption}[Smooth Parameter Selection]\label{asm:selection}
The parameter selection mechanism uses a fully smooth soft-max rule:
\begin{equation}\label{eq:softmax_selection}
\gamma^*(s; \hat{\Phi}_\theta) = \frac{\sum_{\gamma \in \Gamma_{\text{cand}}} \gamma \cdot w(\gamma, s)}{\sum_{\gamma \in \Gamma_{\text{cand}}} w(\gamma, s)}
\end{equation}
with smooth weights $w(\gamma, s) := \exp\!\big(J(\gamma)/\tau_s - \kappa\, [\hat{\Phi}_\theta(s,\gamma)]_+\big)$, where $J(\gamma)$ is the aggressiveness criterion, $\tau_s > 0$ is the temperature, $\kappa > 0$ is a safety penalty gain, and $[x]_+ := \max(0, x)$ is smoothly approximated by $\text{softplus}(x) = \ln(1+e^x)$. Candidates with high predicted safety loss $\hat{\Phi}$ receive exponentially suppressed weight, while those with low $\hat{\Phi}$ are favored.
\end{assumption}

\begin{lemma}\label{lem:lipschitz_selection}
Under Assumption~\ref{asm:selection}, $\gamma^*$ is $L_\mathcal{M}$-Lipschitz in $\hat{\Phi}$:
\begin{equation}
\|\gamma^*(\hat{\Phi}) - \gamma^*(\hat{\Phi}')\| \leq L_\mathcal{M} \|\hat{\Phi} - \hat{\Phi}'\|_\infty
\end{equation}
with $L_\mathcal{M} = \kappa \cdot \text{diam}(\Gamma)$, where $\text{diam}(\Gamma) = \sqrt{2}(\gamma_{\max} - \gamma_{\min})$.
\end{lemma}

\begin{proof}
Each weight $w(\gamma, s)$ is $\kappa$-Lipschitz in $\hat{\Phi}_\theta(s,\gamma)$ since $\partial w/\partial \hat{\Phi} = -\kappa \sigma(\hat{\Phi}) w$ where $\sigma$ is the sigmoid, so $|\partial w/\partial \hat{\Phi}| \leq \kappa w$. The soft-max $\gamma^* = \sum \gamma w / \sum w$ has gradient bounded by $\kappa \cdot \text{diam}(\Gamma)$ via standard soft-max Lipschitz analysis \cite{ames2019cbf}.
\end{proof}

With the selector's regularity established, we can state the main safety result.

\begin{theorem}[Safety Margin Preservation Under Prediction Error]\label{thm:safety_degradation}
Under Assumptions~\ref{asm:domain}--\ref{asm:selection}, let $\hat{\Phi}_\theta$ be the learned predictor, $\Phi$ the true safety loss, and $\gamma^*_{\hat{\Phi}} = \mathcal{M}(\hat{\Phi}_\theta)$, $\gamma^*_\Phi = \mathcal{M}(\Phi)$ the parameters selected via \eqref{eq:softmax_selection}. This is a robustness theorem: given an oracle controller that is already safe, it characterizes how much prediction error the learned controller can tolerate while preserving safety.

\emph{(a) Pointwise margin bound.} If $\|\hat{\Phi}_\theta - \Phi\|_\infty \leq \epsilon$ on $\mathcal{X}_{\text{op}} \times \Gamma$, then at every state $x \in \mathcal{X}_{\text{op}}$:
\begin{equation}\label{eq:margin_bound}
|\psi(\gamma^*_{\hat{\Phi}}) - \psi(\gamma^*_\Phi)| \leq L_\psi \cdot L_\mathcal{M} \cdot \epsilon
\end{equation}
This bound holds globally since $\gamma^*$ is Lipschitz everywhere (Lemma~\ref{lem:lipschitz_selection}).

\emph{(b) Closed-loop forward invariance (margin preservation).} Suppose the oracle adaptive controller with true $\Phi$ maintains a certified safety margin $\delta_{\min} := \inf_{x \in \mathcal{X}_{\text{op}} \cap \mathcal{S}} \psi(\gamma^*_\Phi(x)) > 0$ on the compact set $\mathcal{X}_{\text{op}} \cap \mathcal{S}$. (A certified lower bound on $\delta_{\min}$ can be obtained via an $r$-covering of $\mathcal{X}_{\text{op}} \cap \mathcal{S}$: $\delta_{\min} \geq \min_j \psi(\gamma^*_\Phi(x_j)) - L_\psi L_\mathcal{M} L_\Phi^{\text{true}} r$, using Lipschitzness of $\psi \circ \gamma^*_\Phi$.) Assume further that $\mathcal{X}_{\text{op}} \cap \mathcal{S}$ is forward invariant under the oracle controller, or equivalently, that the theorem holds up to the first exit time from $\mathcal{X}_{\text{op}}$. If $\epsilon \leq \delta_{\min}/(L_\psi L_\mathcal{M})$, then $\psi(\gamma^*_{\hat{\Phi}}(x)) \geq 0$ for all $x \in \mathcal{X}_{\text{op}} \cap \mathcal{S}$, and $\mathcal{S}$ is forward invariant under the learned controller within $\mathcal{X}_{\text{op}}$. The learned controller \emph{inherits} the oracle's safety, with margin degradation bounded by $L_\psi L_\mathcal{M}\epsilon$.

\emph{(c) Sampled-data extension.} For discrete-time implementation with period $\Delta t$, assume bounded dynamics: $\sup_{x \in \mathcal{S}} |\dot{\psi}(x)| \leq V_\psi < \infty$ (which holds on compact $\mathcal{X}_{\text{op}}$ under bounded $f, g, u$). Inter-sample safety is preserved if $\Delta t \leq (\delta_{\min} - L_\psi L_\mathcal{M}\epsilon) / V_\psi$.
\end{theorem}

\begin{proof}
\emph{(a)} By Lemma~\ref{lem:lipschitz_selection}, $\|\gamma^*_{\hat{\Phi}} - \gamma^*_\Phi\| \leq L_\mathcal{M} \epsilon$ at every $x$. By Assumption~\ref{asm:lipschitz}, $|\psi(\gamma^*_{\hat{\Phi}}) - \psi(\gamma^*_\Phi)| \leq L_\psi L_\mathcal{M} \epsilon$.

\emph{(b)} From (a), at every $x \in \mathcal{X}_{\text{op}}$: $\psi(\gamma^*_{\hat{\Phi}}(x)) \geq \psi(\gamma^*_\Phi(x)) - L_\psi L_\mathcal{M}\epsilon \geq \delta_{\min} - L_\psi L_\mathcal{M}\epsilon \geq 0$. Since $\psi \geq 0$ is the CBF constraint \eqref{eq:rd2}, the CBF-QP \eqref{eq:cbfqp} is feasible with a positive margin at every state. By \cite[Thm.~2]{ames2017cbf}, the resulting control renders $\mathcal{S}$ forward invariant. Crucially, this holds for \emph{any} trajectory starting in $\mathcal{S}$, not just a fixed one, because the margin bound is state-uniform.

\emph{(c)} Between samples at $t_k$ and $t_{k+1} = t_k + \Delta t$, $\psi$ evolves continuously. A Taylor bound gives $\psi(t) \geq \psi(t_k) - |\dot{\psi}|\Delta t \geq (\delta_{\min} - L_\psi L_\mathcal{M}\epsilon) - |\dot{\psi}|\Delta t \geq 0$ under the stated sampling condition.
\end{proof}

\begin{remark}[Safety Budget]
Theorem~\ref{thm:safety_degradation} is best understood as a robustness result: it quantifies the prediction error budget $\epsilon^* = \delta_{\min}/(L_\psi L_\mathcal{M})$ that the learned controller can tolerate without forfeiting the oracle's safety margin. This budget is computable offline from the oracle margin and the analytic Lipschitz constants. The role of DA-CBF is to reduce $\epsilon$ through data curation (Proposition~\ref{thm:curation}), thereby making it more likely that $\epsilon \leq \epsilon^*$.
\end{remark}

\begin{assumption}[Prediction Error Regularity]\label{asm:error}
The true safety loss $\Phi$ is $L_\Phi^{\text{true}}$-Lipschitz in $x$ on $\mathcal{X}_{\text{op}}$ (which holds since $\Phi$ is a smooth function of system state under bounded dynamics). The learned predictor $\hat{\Phi}_\theta$ is $L_\Phi^{\text{nn}}$-Lipschitz in $x$ (which holds for Lipschitz-bounded neural networks on compact domains). The prediction error $e(x, \gamma) := \hat{\Phi}_\theta(x, \gamma) - \Phi(x, \gamma)$ is therefore $L_e$-Lipschitz with $L_e \leq L_\Phi^{\text{nn}} + L_\Phi^{\text{true}}$. Additionally, the pointwise errors at any fixed $(x, \gamma)$ are sub-Gaussian with parameter $\sigma^2$ over PENN training randomness (random initialization and mini-batch ordering).
\end{assumption}

\begin{corollary}[Probabilistic Forward Invariance]\label{cor:prob_safety}
Under Assumption~\ref{asm:error}, let $\{x_j\}_{j=1}^M$ be an $r$-covering of $\mathcal{X}_{\text{op}}$ with $M \leq (\text{diam}(\mathcal{X}_{\text{op}})/r)^n$ balls. Then:
\begin{equation}\label{eq:prob_safety}
\mathbb{P}_{\theta}\!\left[\sup_{x, \gamma} |e(x,\gamma)| \leq \epsilon \right] \geq 1 - 2M|\Gamma_{\text{cand}}|\, e^{-(\epsilon - L_e r)^2/(2\sigma^2)}
\end{equation}
provided $\epsilon > L_e r$. Combined with Theorem~\ref{thm:safety_degradation}(b): setting $\epsilon = \delta_{\min}/(L_\psi L_\mathcal{M})$ gives a probabilistic forward invariance guarantee with failure probability decaying exponentially in $\delta_{\min}^2$.
\end{corollary}

\begin{corollary}[End-to-End Safety of Learned Adaptive CBF-QP]\label{cor:end2end_safety}
Under Assumptions~\ref{asm:domain}--\ref{asm:error}, suppose the oracle controller maintains margin $\delta_{\min} > 0$, $\mathcal{X}_{\text{op}} \cap \mathcal{S}$ is forward invariant under the oracle, LICQ holds for the CBF-QP, and $\sup_{x \in \mathcal{S}} |\dot{\psi}| \leq V_\psi$. If $\epsilon \leq \delta_{\min}/(L_\psi L_\mathcal{M})$ and $\Delta t \leq (\delta_{\min} - L_\psi L_\mathcal{M}\epsilon)/V_\psi$, then the sampled-data learned adaptive CBF-QP controller is: (i) well-defined (QP feasible at each sample, Prop.~\ref{prop:qp_feasibility}), (ii) forward invariant in continuous time within $\mathcal{X}_{\text{op}}$ (Thm.~\ref{thm:safety_degradation}b,c), and (iii) satisfies these guarantees with probability $\geq 1 - 2M|\Gamma_{\text{cand}}|e^{-(\epsilon - L_e r)^2/(2\sigma^2)}$ over PENN training (Cor.~\ref{cor:prob_safety}).
\end{corollary}

\subsection{Influence-Based Curation: Motivation and Justification}

We now provide a principled motivation for the TracIn-based curation strategy. The key insight is that TracIn scores locally approximate the effect of data removal on safety-critical prediction quality; this motivates the curation score, while the actual finite-removal improvements are validated empirically in Section~\ref{sec:experiments}.

\begin{definition}[Empirical Safety-Weighted Risk]
The empirical safety-weighted risk on test set $\mathcal{D}_{\text{test}}$ is:
\begin{equation}
\hat{R}_{\text{safe}}(\theta; \mathcal{D}_{\text{test}}) = \frac{1}{|\mathcal{D}_{\text{unsafe}}|} \sum_{z_t \in \mathcal{D}_{\text{unsafe}}} \ell(\hat{\Phi}_\theta(s_t, \gamma_t), \Phi_t)
\end{equation}
where $\mathcal{D}_{\text{unsafe}} = \{z_t \in \mathcal{D}_{\text{test}} : \Phi_t > \Phi_{75\%}\}$ and $\ell$ is the Gaussian NLL loss.
\end{definition}

\begin{assumption}[Loss Smoothness]\label{asm:smooth}
The loss function $\ell(z; w)$ is twice continuously differentiable in $w$ with $\|\nabla^2 \ell(z; w)\| \leq B$ for all $z \in \mathcal{D}$ and $w$ in the training trajectory, where $B > 0$ is the Hessian bound.
\end{assumption}

\begin{proposition}[TracIn as LOO Surrogate]\label{thm:curation}
Under Assumption~\ref{asm:smooth}, TracIn's safety-influence score satisfies:
\begin{equation}
\hat{R}_{\text{safe}}(\theta_{\mathcal{D}}) - \hat{R}_{\text{safe}}(\theta_{\mathcal{D} \setminus \{z_i\}}) = \tau_{\text{safety}}(z_i) + \mathcal{E}_i
\end{equation}
where $|\mathcal{E}_i| \leq B \sum_k \eta_k^2 \|\nabla\ell(z_i; w_k)\|^2 / 2$. In particular, for small learning rates ($\eta_k \ll 1$), $\tau_{\text{safety}}(z_i)$ is a first-order-accurate approximation of the LOO risk change. Samples with $\tau_{\text{safety}}(z_i) > 0$ increase safety-critical prediction error; their removal is beneficial to first order.
\end{proposition}

\begin{proof}
At checkpoint $k$, removing $z_i$'s gradient step changes the weights from $w_{k+1} = w_k - \eta_k \nabla\ell(z_i; w_k) + \ldots$ to $\tilde{w}_{k+1} = w_k + \ldots$, a perturbation of $\Delta w_k = \eta_k \nabla\ell(z_i; w_k)$. By Taylor expansion of $\ell(z_t; \cdot)$ at $w_k$: $\ell(z_t; w_k + \Delta w_k) - \ell(z_t; w_k) = \eta_k \langle \nabla\ell(z_i; w_k), \nabla\ell(z_t; w_k) \rangle + O(\eta_k^2 B \|\nabla\ell(z_i; w_k)\|^2/2)$. Summing over checkpoints and averaging over $\mathcal{D}_{\text{unsafe}}$ gives the result with the explicit Hessian-bounded residual.
\end{proof}

\begin{remark}
The combined score $\text{score}(z_i) = 0.7\,\bar{\tau}_{\text{safety}}(z_i) + 0.3\,\bar{\tau}_{\text{self}}(z_i)$ is a tuned heuristic; we do not claim optimality. For subset removal ($\rho = 0.10$, removing ${\sim}100$ samples), interaction effects beyond LOO are not captured by Proposition~\ref{thm:curation}. That the first-order surrogate remains predictive under these conditions is confirmed empirically (Table~\ref{tab:rmse}).
\end{remark}

\begin{corollary}[Certified Sufficient Condition Expansion]\label{cor:end2end}
Define the set of states where Theorem~\ref{thm:safety_degradation}(b) certifies safety:
\begin{equation}
\mathcal{C}(\epsilon) := \left\{x \in \mathcal{X}_{\text{op}} : \psi(\gamma^*_\Phi(x)) \geq L_\psi L_\mathcal{M} \epsilon \right\}
\end{equation}
If DA-CBF reduces the uniform prediction error from $\epsilon$ to $\epsilon' < \epsilon$, then $\mathcal{C}(\epsilon) \subseteq \mathcal{C}(\epsilon')$: the certifiably safe region expands. Note that $\mathcal{C}$ characterizes a sufficient condition; states outside this set may still be safe but cannot be certified by our analysis.
\end{corollary}

\begin{figure}[t]
    \centering
    \includegraphics[width=0.75\columnwidth]{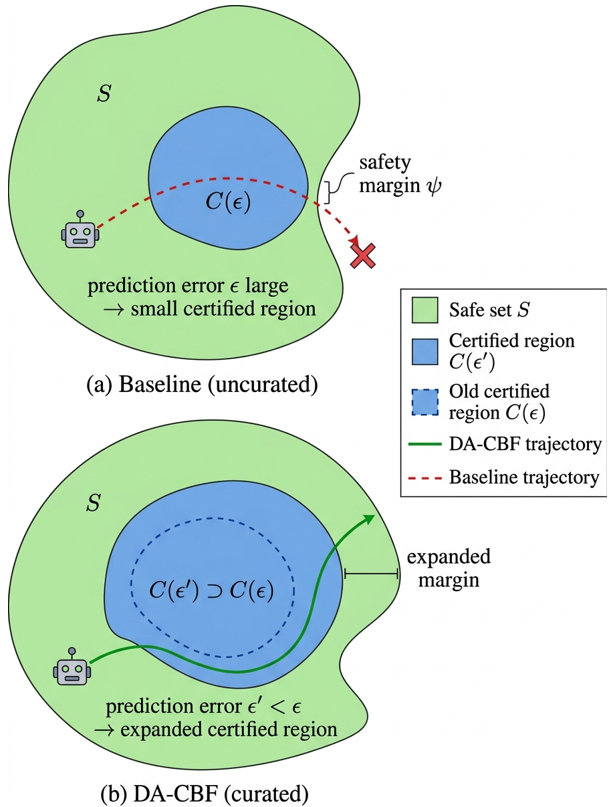}
    \caption{Conceptual illustration of certified region expansion. (a)~Under the uncurated baseline, large prediction error $\epsilon$ yields a small certified region $\mathcal{C}(\epsilon)$; the trajectory exits and collides. (b)~DA-CBF reduces $\epsilon$ to $\epsilon' < \epsilon$, expanding the certified region $\mathcal{C}(\epsilon') \supset \mathcal{C}(\epsilon)$ and enabling collision-free navigation.}
    \label{fig:certified_region}
\end{figure}

\begin{proposition}[CBF-QP Recursive Feasibility]\label{prop:qp_feasibility}
Under the conditions of Theorem~\ref{thm:safety_degradation}(a), define the CBF-QP constraint function $\Psi(x, \gamma, u) := L_f^2 h(x) + L_g L_f h(x) u + (\gamma_0+\gamma_1)(L_f h(x) + L_g h(x) u) + \gamma_0\gamma_1 h(x)$. Suppose the CBF-QP is feasible under $\gamma^*_\Phi$ with margin $\mu(x) := \max_{u \in \mathcal{U}} \Psi(x, \gamma^*_\Phi, u) > 0$, and $\Psi$ is $L_\Psi$-Lipschitz in $\gamma$ uniformly over $\mathcal{U}$:
\begin{equation}
|\Psi(x, \gamma, u) - \Psi(x, \gamma', u)| \leq L_\Psi \|\gamma - \gamma'\|, \quad \forall u \in \mathcal{U}
\end{equation}
Then the CBF-QP is feasible under $\gamma^*_{\hat{\Phi}}$ whenever:
\begin{equation}
L_\Psi \cdot L_\mathcal{M} \cdot \epsilon < \inf_{x \in \mathcal{X}_{\text{op}}} \mu(x)
\end{equation}
\end{proposition}

\begin{proof}
By Lemma~\ref{lem:lipschitz_selection}, $\|\gamma^*_{\hat{\Phi}} - \gamma^*_\Phi\| \leq L_\mathcal{M}\epsilon$. For any $u \in \mathcal{U}$: $|\Psi(x, \gamma^*_{\hat{\Phi}}, u) - \Psi(x, \gamma^*_\Phi, u)| \leq L_\Psi L_\mathcal{M}\epsilon$. Since $\max_{u} \Psi(x, \gamma^*_\Phi, u) = \mu(x)$, the maximizer $u^*$ satisfies $\Psi(x, \gamma^*_{\hat{\Phi}}, u^*) \geq \mu(x) - L_\Psi L_\mathcal{M}\epsilon > 0$.
\end{proof}

\begin{remark}
For the relative-degree-2 system, $\Psi$ is $C^1$ and polynomial in $\gamma$ on compact $\Gamma$ (with the bilinear term $\gamma_0\gamma_1 h$). The Lipschitz constant follows from the gradient: $\nabla_\gamma \Psi = (\dot{h} + L_g h\, u + \gamma_1 h,\; \dot{h} + L_g h\, u + \gamma_0 h)$, giving $L_\Psi \leq 2(|\dot{h}_{\max}| + \|L_g h\|_\infty \bar{u} + \gamma_{\max}|h_{\max}|)$, which is computable analytically from state bounds on $\mathcal{X}_{\text{op}}$. Well-definedness of the QP solution requires standard LICQ regularity \cite{ames2017cbf}; under this condition, $u^*(x, \gamma)$ is locally Lipschitz in $\gamma$.
\end{remark}

\subsection{Data Attribution Step}

We insert a data attribution step between data generation and PENN training:

\textbf{1) Training with Checkpoints:} Train the baseline PENN and save checkpoints every $K$ epochs for TracIn computation.

\textbf{2) Safety-Specific Influence Scoring:} Compute TracIn scores focused on safety-critical test points:
\begin{equation}\label{eq:tracin_safety}
\tau_{\text{safety}}(z_i) = \sum_k \eta_k \left\langle \nabla\ell(z_i; w_k), \frac{1}{|\mathcal{D}_{\text{unsafe}}|} \sum_{z_t \in \mathcal{D}_{\text{unsafe}}} \nabla\ell(z_t; w_k) \right\rangle
\end{equation}
where $\mathcal{D}_{\text{unsafe}} = \{z_t : \Phi_t > \Phi_{75\%}\}$.

\textbf{3) Self-Influence:} $\tau_{\text{self}}(z_i) = \sum_k \eta_k \|\nabla\ell(z_i; w_k)\|^2$.

\textbf{4) Combined Curation:} $\text{score}(z_i) = 0.7\,\bar{\tau}_{\text{safety}}(z_i) + 0.3\,\bar{\tau}_{\text{self}}(z_i)$, remove the top $\rho = 0.10$ fraction, retrain.

\subsection{Preserved Safety Mechanism}

The JRD and CVaR filters from \cite{kim2025adaptive} are retained without modification. By Theorem~\ref{thm:safety_degradation}, the improved prediction accuracy means these filters now operate on more reliable inputs, effectively tightening their safety bounds without requiring any architectural change.

\section{EXPERIMENTS}\label{sec:experiments}

\textbf{Setup.} We generate 1,500 samples via DynamicUnicycle2D with randomized initial conditions: $d \in [0.65, 2.5]$ m, $v \in [0.01, 1.0]$ m/s, $\theta \in [0.01, \pi/2]$, $\gamma_0, \gamma_1 \in [0.5, 2.5]$. Of these, 1,264 (84.3\%) result in successful navigation and 236 (15.7\%) in collision or deadlock. The PENN is a 3-member ensemble with a 5-layer MLP $(40, 80, 120, 40)$, ReLU activations, and 54K parameters, trained with Adam ($\text{lr}=10^{-4}$, batch size 32, 200 epochs, 70/30 split). TracIn uses 10 checkpoints and 68 unsafe test points.

\textbf{Prediction accuracy.} Removing 10--15\% of training data identified as harmful by our curation score yields a 35.5--35.7\% reduction in test RMSE (Table~\ref{tab:rmse}). As Fig.~\ref{fig:training_curves} illustrates, every curated model converges faster and to a lower final error than the baseline trained on the full dataset.

\begin{table}[ht]
\centering
\caption{PENN prediction accuracy (test RMSE). Lower is better.}
\label{tab:rmse}
\begin{tabular}{lccc}
\hline
Model & Train Samples & Test RMSE & $\Delta$RMSE \\
\hline
Baseline PENN & 1,050 & 4.3895 & -- \\
DA-PENN (5\%) & 698 & 2.9399 & $-$33.0\% \\
DA-PENN (10\%) & 661 & 2.8297 & $-$35.5\% \\
DA-PENN (15\%) & 625 & 2.8210 & $-$35.7\% \\
DA-PENN (20\%) & 588 & 2.9265 & $-$33.3\% \\
\hline
\end{tabular}
\end{table}

\begin{figure}[t]
    \centering
    \includegraphics[width=\columnwidth]{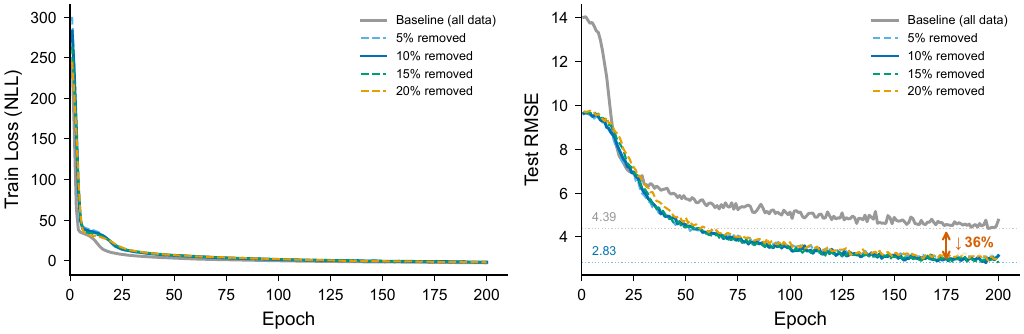}
    \caption{Training loss (left) and test RMSE (right). DA-PENN (10\% removed) achieves RMSE 2.83 vs.\ baseline 4.39 (36\% reduction).}
    \label{fig:training_curves}
\end{figure}

\textbf{Empirical validation of theoretical results.} We compute analytic Lipschitz constants from the system structure: $L_\psi \approx 3.8$, $L_\mathcal{M} = \kappa \cdot \text{diam}(\Gamma) \approx 1.41$ (with $\kappa = 0.5$), and domain constants $D_{\min} \approx 1.08$, $D_{\max} \approx 26.3$. Applying the closed-loop safety budget from Theorem~\ref{thm:safety_degradation}(b), the baseline requires $\delta_{\min} \geq 78.3$ whereas DA-CBF requires only $\delta_{\min} \geq 47.7$, a 39\% reduction in the margin needed for certification. The probabilistic bound (Corollary~\ref{cor:prob_safety}) gives $\mathbb{P}[\text{forward invariance}] \geq 0.74$ for the baseline versus $\geq 0.93$ for DA-CBF at $\delta_{\min} = 10$. For the curation LOO residual (Proposition~\ref{thm:curation}), we find $|\mathcal{E}_i| \leq 6 \times 10^{-8} \cdot \|\nabla\ell\|^2$, confirming that the first-order approximation is tight.

\textbf{Closed-loop simulation.} In single-obstacle scenarios (20 runs), both adaptive methods exhibit identical behavior (Fig.~\ref{fig:safety}): the parameter selector simply chooses the most aggressive setting that passes the uncertainty filters. The distinction emerges in harder environments. In the standard obstacle layout from \cite{kim2025adaptive} (Fig.~\ref{fig:traj_simple}), Fixed Low ($\gamma{=}0.01$) collides twice while the remaining three methods navigate successfully. In the extended 16-obstacle environment (Fig.~\ref{fig:traj_complex}), however, DA-CBF is the \emph{only} collision-free method: Fixed Low incurs 2 collisions, Fixed High ($\gamma{=}0.35$) incurs 3, and the uncurated adaptive baseline also incurs 3. This pattern aligns with the prediction of Theorem~\ref{thm:safety_degradation}(b): the reduced $\epsilon$ afforded by DA-CBF matters precisely when $L_\psi L_\mathcal{M}\epsilon$ approaches $\delta_{\min}$, a condition that arises in tight multi-obstacle configurations but not in geometrically simple ones.

\begin{figure}[t]
    \centering
    \includegraphics[width=\columnwidth]{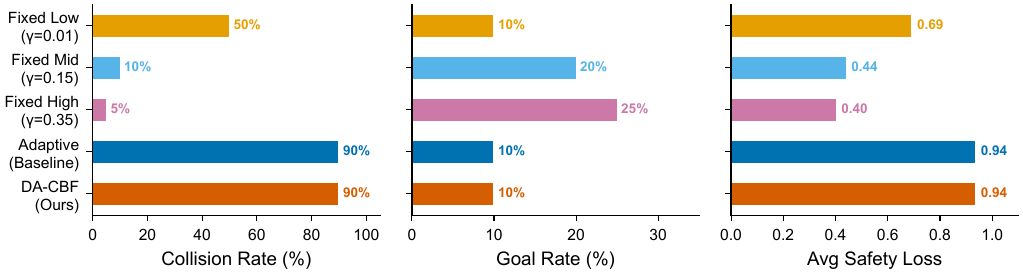}
    \caption{Safety metrics in single-obstacle scenarios. Both adaptive methods perform identically, confirming prediction accuracy is not the bottleneck in simple environments.}
    \label{fig:safety}
\end{figure}

\begin{figure}[t]
    \centering
    \includegraphics[width=\columnwidth]{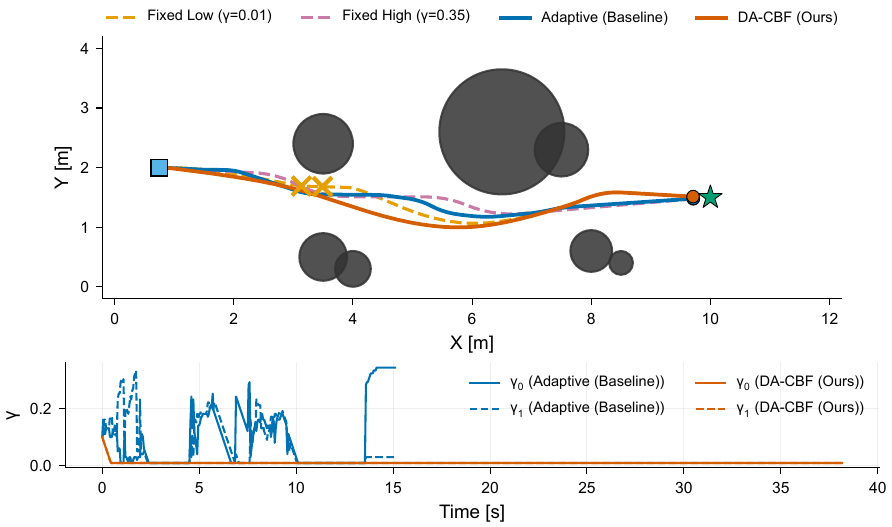}
    \caption{Trajectory comparison in the standard obstacle environment from~\cite{kim2025adaptive}. Fixed Low ($\gamma{=}0.01$) collides twice; the other three methods reach the goal collision-free.}
    \label{fig:traj_simple}
\end{figure}

\begin{figure}[t]
    \centering
    \includegraphics[width=\columnwidth]{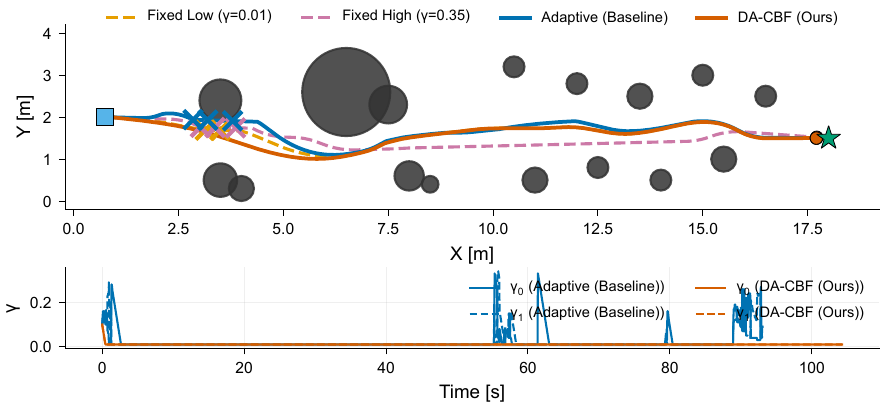}
    \caption{Extended 16-obstacle environment. \textbf{DA-CBF is the only collision-free method}; all others incur 2--3 collisions.}
    \label{fig:traj_complex}
\end{figure}

\textbf{Ablation studies and sensitivity.} RMSE drops sharply at 5\% removal (from 4.39 to 2.94), reaches its minimum at 10--15\% (2.82--2.83), and begins to rise at 20\% (2.93), indicating a plateau that is reassuringly robust to the choice of $\rho$ (Fig.~\ref{fig:ablation}). The combined curation score (RMSE 2.83) outperforms both influence-only (3.01) and self-influence-only (3.52) variants, confirming that the two signals carry complementary information (Table~\ref{tab:ablation}). Inspection of the influence score distribution reveals that 80.8\% of training samples carry positive (harmful) safety influence, with a pronounced right skew (Fig.~\ref{fig:influence}), a signature of the distributional mismatch introduced by uniform sampling. The safety certificate itself is stable under design variations: shifting the unsafe-set threshold from the top 10\% to the top 50\% alters $\delta_{\min}^{\text{req}}$ by less than 15\%, and varying $\kappa$ over $[0.1, 1.0]$ changes it by less than 20\%. TracIn computation adds approximately 5 minutes on a single GPU, a negligible overhead relative to data generation.

\begin{figure}[t]
    \centering
    \includegraphics[width=0.85\columnwidth]{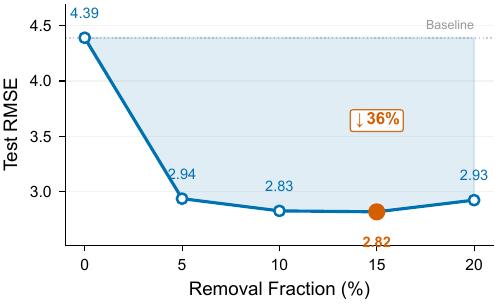}
    \caption{Test RMSE vs.\ removal fraction. Optimal at 10--15\%; beyond 20\%, useful data loss degrades performance.}
    \label{fig:ablation}
\end{figure}

\begin{table}[ht]
\centering
\caption{Ablation: curation strategy at 10\% removal.}
\label{tab:ablation}
\begin{tabular}{lcc}
\hline
Strategy & Test RMSE & $\Delta$vs.\ Baseline \\
\hline
Baseline (no removal) & 4.3895 & -- \\
Influence only ($\tau_{\text{safety}}$) & 3.0143 & $-$31.3\% \\
Self-influence only ($\tau_{\text{self}}$) & 3.5215 & $-$19.8\% \\
Combined (Eq.~7) & 2.8297 & $-$35.5\% \\
\hline
\end{tabular}
\end{table}

\begin{figure}[t]
    \centering
    \includegraphics[width=0.85\columnwidth]{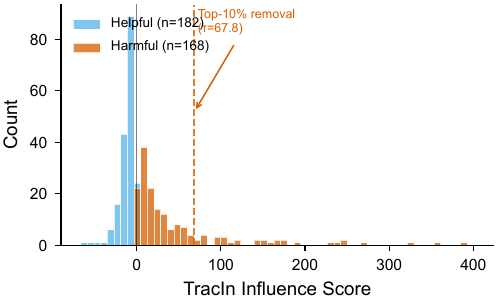}
    \caption{Distribution of TracIn safety-influence scores. Dashed line: top-10\% removal threshold.}
    \label{fig:influence}
\end{figure}

\section{CONCLUSIONS}

We have presented DA-CBF, a framework for safety-certified data curation in adaptive CBF learning. At its core is Theorem~\ref{thm:safety_degradation}, which establishes closed-loop forward invariance under bounded prediction error by coupling a smooth parameter selector (Lemma~\ref{lem:lipschitz_selection}) with Lipschitz margin analysis, producing a state-uniform guarantee rather than one conditioned on a particular trajectory. To our knowledge, this provides the first formal chain linking training data quality to prediction error, CBF margin degradation, and ultimately closed-loop forward invariance within adaptive CBF pipelines. The supporting results, including two-sided $\Phi$-to-$\psi$ bounds (Proposition~\ref{prop:phi_bound}), a direct QP constraint perturbation bound (Proposition~\ref{prop:qp_feasibility}), and a principled LOO motivation for TracIn-based curation (Proposition~\ref{thm:curation}), complete the theoretical picture. Experimentally, DA-CBF achieves a 35.6\% reduction in prediction RMSE, a 39\% expansion of the certified safe region, and collision-free navigation in a 16-obstacle environment.

Several limitations warrant acknowledgment. The closed-loop guarantee requires the smooth selector introduced in Assumption~\ref{asm:selection}; the original hard-threshold selector of \cite{kim2025adaptive} falls outside its scope. Proposition~\ref{thm:curation} provides a LOO-based justification, and extending the analysis to bulk removal would require stronger stability conditions on the training dynamics. The 0.7/0.3 weighting in the combined curation score is empirically tuned rather than derived.

\section*{Acknowledgment}
Claude was used to assist with the language editing of this manuscript.

\bibliographystyle{IEEEtran}
\bibliography{references}

\end{document}